\RequirePackage{fix-cm}
\documentclass[twocolumn]{svjour3}          % twocolumn
\smartqed  % flush right qed marks, e.g. at end of proof
\usepackage{graphicx}
\usepackage{amsmath}
\usepackage{amssymb}
\usepackage{tikz}
\usepackage{booktabs}
\usepackage{subcaption}
\usetikzlibrary{calc,patterns,decorations.pathmorphing,decorations.markings,positioning,arrows}
\usepackage{pgfplots}
\pgfplotsset{compat=1.15}
\usepackage{graphicx}

\DeclareCaptionSubType{subfigure}
\AtBeginDocument{% trial and error
  \let\oldsubsubfigure=\subsubfigure
  \renewcommand{\subsubfigure}{\expandafter\def\csname @captype\endcsname{subfigure}%
    \oldsubsubfigure}%
}

\newsavebox{\subsubfloatbox}% probably overkill
\newcommand{\subsubfloat}[2][\empty]% #1 = caption (optional), #2 = image
{\bgroup
  \captionsetup[subsubfigure]{font=footnotesize}%
  \savebox\subsubfloatbox{#2}%
  \begin{subsubfigure}[t]{\wd\subsubfloatbox}
    \usebox\subsubfloatbox
    \ifx\empty#1\relax
      \stepcounter{subsubfigure}%
    \else
      \caption{#1}%
    \fi
  \end{subsubfigure}%
\egroup}

\counterwithin*{theorem}{section}

\usepackage{ragged2e}
\usepackage{enumerate}

\usepackage{etoolbox}
\newcommand*{\affaddr}[1]{#1} % No op here. Customize it for different styles.
\newcommand*{\affmark}[1][*]{\textsuperscript{#1}}

\begin{document}

\title{Geometric Stability Estimates For 3D-Object Encryption Through Permutations and Rotations %\thanks{Grants or other notes
}

\titlerunning{Geometric Estimates for 3D-Object  Ciphers}        % if too long for running head
\author{M.H. Annaby\affmark[1,$\ast$] \and M.E. Mahmoud\affmark[1] \and H.A. Abdusalam\affmark[1] \and H.A. Ayad\affmark[1] \and M.A. Rushdi\affmark[2]}

\authorrunning{M.H. Annaby  et al. } % if too long for running head

\institute{M.H. Annaby \at
             % Tel.: +123-45-678910\\
            \email{mhannaby@sci.cu.edu.eg}        %  \\
              \and
              \affaddr{\affmark[$\ast$]Corresponding author}\\
              \affaddr{\affmark[1]Department of Mathematics, Faculty of Science, Cairo University, Giza, 12613, Egypt}
              \affaddr{\affmark[2]Department of Biomedical Engineering and Systems, Faculty of Engineering, Cairo University, Giza, 12613, Egypt}
}

\date{Received: date / Accepted: date}

\maketitle

\begin{abstract}
We compute precise estimates for dimensions of 3D-encryption techniques of 3D-point clouds which use permutations and rigid body motion, in which geometric stability is to be guaranteed. Few attempts are made in this direction. An attempt is established using the notions of dimensional and spatial stability by  Jolfaei et al. (2015), who also proposed a 3D object encryption algorithm, claiming that it preserves dimensional and spatial stability. However, as we mathematically prove neither the algorithm, nor the  associated estimates are correct. We introduce more rigorous definitions of the geometric stability of such 3D data encryption algorithms, followed by dimensionality measures. 
\end{abstract}

%\maketitle

%%%%%%%%%%%%%%%%%%%%%%%%%%%%%%%%%%%%%%%%%%%%%%%%%%%%%%%%%%%%%%%%%%%%%%%%%%%%%%%%%%%%%%%%%%%%%%%%%%%%%%%%%%%%%%%%%%%%%%%%%%%%%%%%%%%%%%%%%%%%%%%%%%%%%%%%%%%%%%%%%%%%%
%% Introduction %%
%%%%%%%%%%%%%%%%%%%
\section{Introduction}
\label{section:intro}
The 3D object structures are widely utilized in various applications. See e.g.  \cite{lee2017fundamentals} for 3D modeling and 3D printing,  \cite{li2005collaborative} for computer-aided design,  \cite{popkonstantinovic20123d} for 3D-animation, \cite{gilbert2004models,vernon2002benefits} for interactive anatomical modeling and education,  \cite{pal2001easy} for prototyping and manufacturing,  \cite{scheenstra2005survey} for face recognition, \cite{gomez2015intelligent} for surveillance systems; and the surveys \cite{scheenstra2005survey,thomas2012survey}. The wide applicability and potential vulnerabilities of the 3D object models raised concerns of security and access control. Unauthorized access to 3D models leads to significant security breaches or business losses. Therefore, there has been a growing demand for 3D object encryption algorithms of high security strength, integrity, and robustness to attacks.  While numerous 1D and 2D ciphers have been designed, particularly for digital images, the 3D object encryption algorithms are still quite limited. 

Several design aspects should be considered in the construction of 3D ciphers. A primary aspect is the design of a cipher with reasonable efficiency. Another design aspect of a 3D cipher is its geometric stability. Moreover, a 3D cipher should demonstrate robustness against statistical, differential and chosen-plaintext attacks. Cipher robustness is typically strengthened by using chaotic maps which are dynamical systems that take some initial conditions and control parameters to produce chaotic sequences. 

In addition to implementing chaotic permutations, several schemes used transformations of motion of rigid body, as well as shuffling coordinates of the plaintext with or without mixing the plaintext with chaotically selected 3D-objects, see e.g. \cite{jia2019encryption,jin20163d,jin20173d,jolfaei20143d}. Nevertheless, both shuffling coordinates or the transformations may affect both the correctness of the algorithm and/or the geometric stability of the ciphertext if they are not carefully established. In \cite{jolfaei20143d}, Jolfaei et al. introduced the notions of dimensional and spatial stability of 3D-point cloud ciphers. Furthermore, they constructed a cipher that is based on permutations of coordinates and localized rotations, and claimed that this cipher would be dimensionally and spatially stable. However, as we show through counterexamples and mathematical proofs, these stability notions are not consistent, and the associated cipher is not guaranteed to work correctly. Indeed, the major reason for the geometric instability of the cipher of \cite{jolfaei20143d} is the rotation and shuffling the coordinates of the plaintext with the coordinates of randomly created points. As we indicate in Section \ref{section:Geometric_Stability} such a process leads to a geometric instability.  

Section \ref{section:Related_Works} outlines the algorithm of the 3D object encryption  of Jolfaei et al. \cite{jolfaei20143d}, including a counterexample that proves that the algorithm of \cite{jolfaei20143d} is not correct. The notions of dimensional and spatial stability  are discussed in Section \ref{section:Geometric_Stability}. We  point out to the faults of the mathematical proofs given in \cite{jolfaei20143d} and another counterexample is given.  We revise the notion of the geometric stability and derive exact estimates for ciphertexts under reotations and shuffling coordinates. 

It is worthwhile to mention that, while our approach is estimating dimensionality connected with the algorithm of \cite{jolfaei20143d}, it exhibits a general treatment of measuring geometric stability of encryption algorithms that are based on shuffling coordinates and the motion of rigid body transformations.

%%%%%%%%%%%%%%%%%%%%%%%%%%%%%%%%%%%%%%%%%%%%%%%%%%%%%%%%%%%%%%%%%%%%%%%%%%%%%%%%%%%%%%%%%%%%%%%%%%%%%%%%%%%%%%%%%%%%%%%%%%%%%%%%%%%%%%%%%%%%%%%%%%%%%%%%%%%%%%%%%%%%%
%%%%%%%%%%%%%%%%%%%%%%%%%%%%%%%%%%%%%%%%%%%%%%%%%%%%%%%%%%%%%%%%%%%%%%%%%%%%%%%%%%%%%%%%%%%%%%%%%%%%%%%%%%%%%%%%%%%%%%%%%%%%%%%%%%%%%%%%%%%%%%%%%%%%%%%%%%%%%%%%%%%%%
%% Related Work %%
%%%%%%%%%%%%%%%%%%%
\section{Cipher Instability} 
\label{section:Related_Works}
In the following, we briefly review the 3D point cloud encryption algorithm of Jolfaei et al. \cite{jolfaei20143d}. Let $K^0 = (k_1^0, k_2^0, k_3^0, k_4^0, k_5^0, k_6^0)^\top,$ where $k^0_i \in [-1, 1],\,1 \leq i \leq 6$ be fixed and $\mathcal{P}_N=\{P^1, P^2, \cdots, P^N\} \subseteq \mathbb{R}^3$ be a given 3D-point cloud, $N \in \mathbb{N} = \{ 1, 2, \cdots \}, \ N \geq 2$ is fixed. We assume that $\mathcal{P}_N$ lies in a bounding sphere of radius $r_p > 0 $ and center $P^0 \in \mathbb{R}^3$, i.e. $\|P^j-P^0\| \leq r_p$. Here, $P^j = ( p_1^j, p_2^j, p_3^j )^\top$, where $A^{\top}$ denotes the transpose of $A$. Using the key $K^0$ and the Chebyshev map, $D > 2$,
\begin{equation}\label{Eq: Chebmap}
k^j_i=\cos(D\cos^{-1}(k^{j-1}_i)), \; 1\leq i \leq 6, \ 1\leq j\leq 2N,
\end{equation}
cf.  \cite{GEISEL1984263,NIANSHENG2011761},
$K^j=(k_1^j, k_2^j, k_3^j, k_4^j, k_5^j, k_6^j)^\top,$ 
where $k^j_i \in [-1, 1],$  are created. Then, random points and angles are generated by
\begin{equation}\label{Eq: Randpseudopts}
O^j_{\upsilon} = P^0 + r_p \cdot \left(k_1^{j+\lfloor \frac{\upsilon}{2}\rfloor N}, k_2^{j+\lfloor \frac{\upsilon}{2}\rfloor N}, k_3^{j+\lfloor \frac{\upsilon}{2}\rfloor N} \right)^{\top},
\end{equation}
\begin{equation}\label{eq:rotAngs}
\Lambda ^j_{\upsilon} = (\alpha^j_{\upsilon 1},\alpha^j_{\upsilon 2},\alpha^j_{\upsilon 3})^{\top}, \;
\alpha_i^{j+\lfloor \frac{\upsilon}{2}\rfloor N} = \lfloor 180^{\circ} k^{j+\lfloor \frac{\upsilon}{2}\rfloor N}_{i+3}\rfloor,
\end{equation}
where $i = 1, 2, 3$, \ $1\leq j\leq N$, and $\lfloor\cdot\rfloor$ denotes the floor function. Here $\upsilon \in \{1, 2\}$ is an index for each round of the cipher. So, $\mathcal{O}^\upsilon_N = \{ O^1_{\upsilon}, \cdots, O^N_{\upsilon}\} \subseteq \mathbb{R}^3$
indicates the set of random points created for a specific round $\upsilon$. Notice that $\|O^j_{\upsilon}-P^0\| \leq \sqrt{3}r_{p}$, $\|\cdot\|$ denotes the Euclidean norm. Let $X$ be the set of the $6N$ coordinates of the points contained in $\mathcal{P}_N \cup \mathcal{O}_N$. If $N>8,$ the set $X$ is split into $\lfloor\frac{N}{8}\rfloor$ subsets, provided that $N\ge8,$ where each subset is randomly permuted to get a new set of points $P^{\prime j}, O^{\prime j}, 1 \leq j \leq N$. The first round cipher of $\mathcal{P}_N$ is obtained via localized rotations of the $P^{\prime j}$ points around the $O^{\prime j}$ points as follows \cite[Eq. (16)]{jolfaei20143d},
\begin{equation}\label{eq:rot}
C_1^{j}= \psi \cdot R^j(\alpha^j_{11},\alpha^j_{12},\alpha^j_{13}) \times [P^{\prime j}-O_1^{\prime j}]+O_1^{\prime j}, 
\end{equation}
$1\leq j\leq N,$ where $\psi \in (0, \frac{1}{9}]$, is a factor to guarantee the geometric stability of ciphertext and $R^j(\alpha^j_{11},\alpha^j_{12},\alpha^j_{13})$ is a 3D rotation matrix. Thus, \cite[p.147]{goldstein:mechanics} $R^j=R^j_1\cdot R^j_2\cdot R^j_3$, where $R^j_1,\, R^j_2,\, R^j_3$ are the rotation matrices about $X, Y, $ and $ Z$ axes with angles $\alpha^j_{1},\alpha^j_{2},\alpha^j_{3}$, respectively.
Another round of encryption is carried out similarly which is dispensable as we see that the information cannot be retrieved after implementing one round. In Section \ref{section:Geometric_Stability}, we investigate the geometric stability of the algorithm of \cite{jolfaei20143d}, but first we show that this algorithm may not work.

A major fault in the aforementioned algorithm of \cite{jolfaei20143d} is the random permutation stage. While the encryption process may work, the decryption process does not  work since, for some permutations, the encryption process simply destroys the data.  Jolfaei et al. \cite{jolfaei20143d} suggest that using several encryption rounds may give a higher level of security at the expense of lower  efficiency \cite[p. 412]{jolfaei20143d}. However, irrespective of the number of encryption rounds, the encryption-decryption process may fail. The following illustrative counterexample indicates that even with one round of simple 2D encryption, the above algorithm does not work.

%%%%%%%%%%%%%%%%%%%%%%%%%%%%%%%%%%%%%%%%%%%
%% counterexample 1 %%
%%%%%%%%%%%%%%%%%%%%%%

\begin{example}\label{counter_ex1}
Assume that we are given a ciphertext $\mathcal{C}_2 = \{ C^1, C^2\}, C^j = \begin{pmatrix} c^j_1 \\ c^j_2  \end{pmatrix} , j= 1, 2 $ and that we are given a key that generates a keystream $k^1_1 = \frac{1}{3}, \ k^1_2 = \frac{-1}{2}, \ k^1_3 = \frac{1}{4}, \ k^1_4 = \frac{-2}{3}, \ k^1_5 = \frac{4}{5}, \ k^1_6 = \frac{1}{3}.$ Let the plaintext be $\mathcal{P}_2 = \{P^1, P^2\}, P^1= \begin{pmatrix} x_1 \\ y_1  \end{pmatrix}, P^2 = \begin{pmatrix} x_2 \\ y_2  \end{pmatrix}, P^0 = \begin{pmatrix} 0 \\ 0  \end{pmatrix}$. Thus, $X = \left\{ x_1, y_1, x_2, y_2, \frac{1}{3}, \frac{-1}{2}, \frac{-2}{3}, \frac{4}{5} \right\}.$ 
Let us form $P^{\prime}_j, O^{\prime}_j, \ j = 1,2,$ to be
\begin{equation*}\label{eq:P_O_prime_CE1}
\begin{matrix}
P^{\prime}_1 = \begin{pmatrix} x_1 \\ y_2  \end{pmatrix}, & 
P^{\prime}_2 = \begin{pmatrix} y_1 \\ \frac{1}{3}  \end{pmatrix}, &
O^{\prime}_1 = \begin{pmatrix} x_2 \\ \frac{4}{5}  \end{pmatrix}, &
O^{\prime}_2 = \begin{pmatrix} \frac{-1}{2} \vspace{1mm} \\ \frac{-2}{3} \end{pmatrix}.
\end{matrix}
\end{equation*}
Thus (\ref{eq:rot}) leads to
\begin{equation}\label{eq:inconsistent_CE1}
\begin{array}{l}
x_1 + \left(\displaystyle\frac{\sqrt{2}}{\psi} -1\right)x_2 - y_2 = \frac{\sqrt{2}}{\psi} c^1_1- \frac{4}{5},\\
x_1 - x_2 + y_2 = \displaystyle\frac{\sqrt{2}}{\psi} \left( c^1_2 - \frac{4}{5} \right) + \frac{4}{5},\\
y_1 = \displaystyle\frac{2}{\sqrt{3} \psi} \left( c^2_2 + \frac{2}{3} \right) - \frac{2 + \sqrt{3}}{2 \sqrt{3}},
\end{array}
\end{equation}
and the decryption problem has no solution. Here $r_p=1.$
\end{example}
\begin{remark}
Certain choices of $\psi$ will also lead to inconsistency. It is worthy to mention that this simple example indicates the mistake of the algorithm although we take only 2D plaintext with two points, i.e. $N=2.$ We do not need to take $N>8$, which can be easily done as well, with more complicated systems. One  noticed that the paper \cite{jolfaei20143d} includes no examples of a complete encrypted-decrypted 3D procedure.
\end{remark}
Before we investigate the issues of geometric stability in the next section, we would like to mention that the set $X$ of the coordinates of $\mathcal{P}_N$ and $\mathcal{O}^\upsilon_N$ constitutes an array and the coordinates might be repeated. So, equations like \cite[Eqs. (8), (9)]{jolfaei20143d} are meaningless.

\begin{figure*}[]
	\centering
	\subfloat[]{\includegraphics[width=1.8in, height=2in]{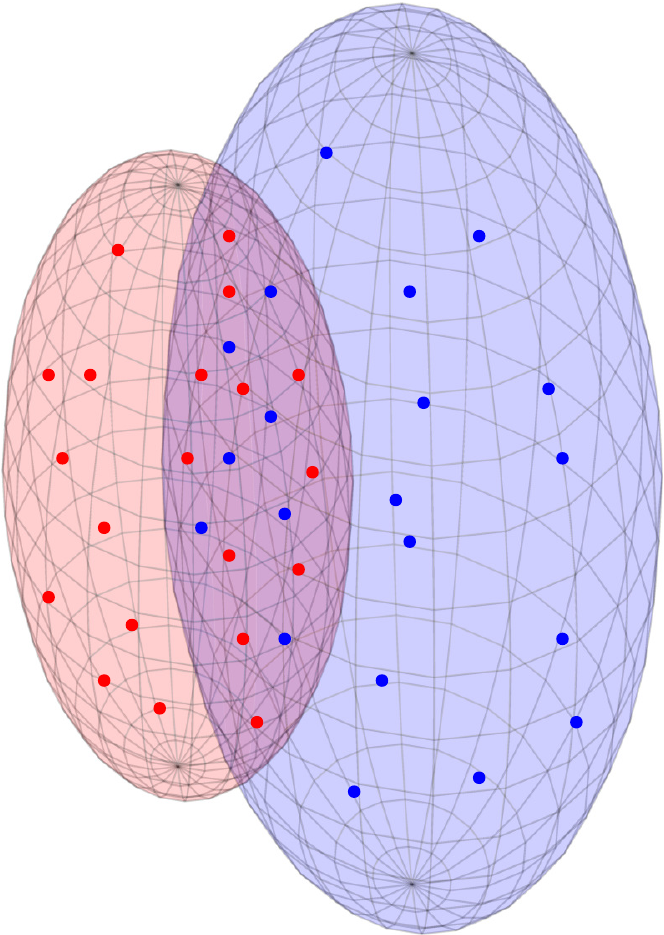}
		\label{fig:GS_case1}}
	\hfil \hfill
	\subfloat[]{\includegraphics[width=1.8in, height=2in]{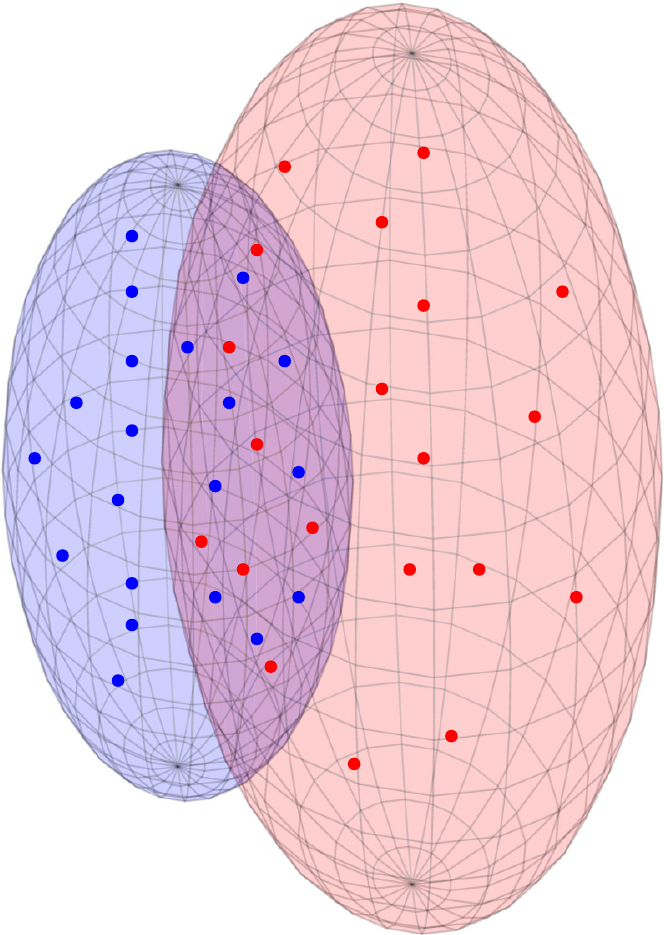}
		\label{fig:GS_case2}}
	\hfil \hfill
	\subfloat[]{\includegraphics[width=1.8in, height=2in]{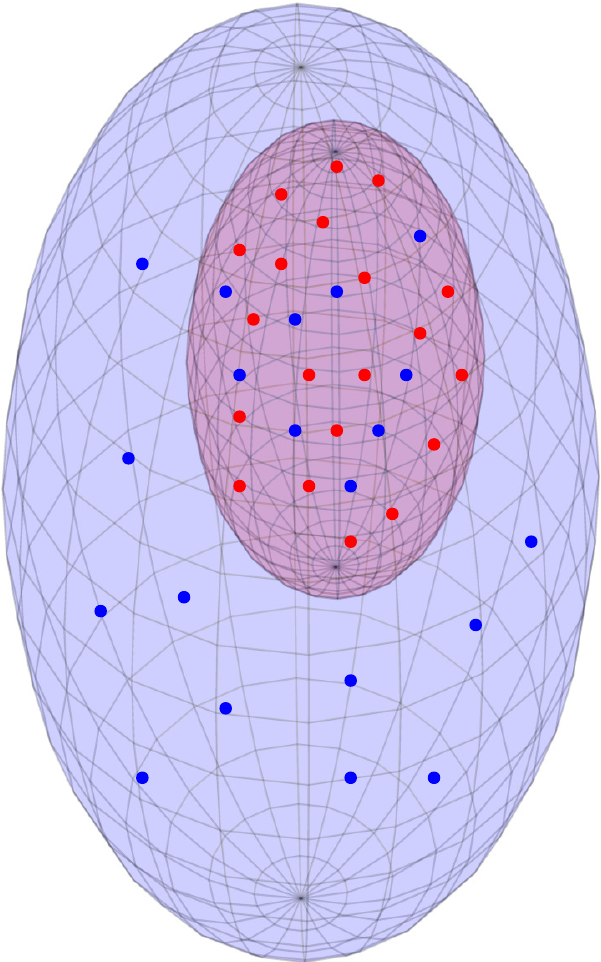}
		\label{fig:GS_case3}}
	\caption{\small{The blue sphere contains the plaintext $\mathcal{P}_{20}$, and the red one contains the ciphertext $\mathcal{C}_{20}$. A cipher that is: (a)  dimensionally  but not spatially stable, (b) neither dimensionally nor spatially stable, (c)  geometrically stable .}}
	\label{fig_geometricStability_3d}
\end{figure*}
%%%%%%%%%%%%%%%%%%%%%%%%%%%%%%%%%%%%%%%%%%%%%
\section{Geometric Stability Dimensions}
%%%%%%%%%%%%%%%%%%%%%%%%%%%%%%%%%%%%%%%%%%%%%
\label{section:Geometric_Stability}
Geometric stability is a key consideration in the design of 3D-point cloud encryption algorithms to avoid exceeding the viewing screen resolution and to avoid collisions between 3D objects. Eluard et al. \cite{eluard2013geometry}, tried to maintain the geometric stability of 3D encryption via constraining the encrypted points to fall within minimum bounding boxes of the plain point clouds. In \cite{jolfaei20143d}, the authors introduced two notions of stability for 3D point cloud encryption algorithms: dimensional stability and spatial stability. These notions are mathematically formulated as follows. Let $C^0$ be the geometric center of the ciphertext 3D-point cloud, and let $r_c$ be the radius of the bounding sphere of the cipher point cloud $\mathcal{C}_N = \{C^1, C^2, \cdots, C^N\}$. Then, $\mathcal{C}_N$ is contained in the ball $\|P-C^0\|\le r_c$. According to \cite{jolfaei20143d}, the cipher is called spatially stable if $\|C^0-P^0\| \leq r_p,$ and it is called dimensionaly stable if $r_p - \|C^0-P^0\| \geq r_c.$ The last inequality is merely $\|C^0-P^0\| \leq r_p - r_c \leq r_p,$ i.e. dimensional stability implies the spatial stability provided that $r_p \geq r_c$.
Since the spatial stability guarantees the occurrence of the encrypted point cloud within the sphere $\|P-P^0\| = r_p$, and the dimensional stability guarantees that the dimensional size of the set $\mathcal{C}_N$ does not exceed that of $\mathcal{P}_N$, we define the geometric stability of a cipher as follows.

\begin{definition}
\label{def:geometric_stable}
Let $\mathcal{P}_N$ be a plain 3D-point cloud and $\mathcal{C}_N$ be the corresponding cipher point cloud under a certain cipher. Let $P^0, C^0, r_p, r_c$ be as described above. The cipher is called dimensionally stable if
\begin{equation}\label{eq:dimensionally_stable}
r_p \geq r_c> 0.
\end{equation}
The cipher is called spatially stable if
\begin{equation}\label{eq:spatially_stable}
0 \leq \|C^0-P^0\| \leq r_p - r_c.
\end{equation}
The cipher is called geometrically stable if it is both dimensionally and spatially stable.
\end{definition} 

Figure \ref{fig_geometricStability_3d} illustrates the geometric stability in three dimensions. Before we investigate the geometric stability of ciphers based on shuffling coordinates and localized rotations, we prove via a counterexample that the algorithm of  \cite{jolfaei20143d} is not geometrically. In fact it is neither dimensionally nor spatially stable by any means.
%%%%%%%%%%%%%%%%%%%%%%%%%%%%%%%%%%%%%
\begin{example}\label{counter_ex2}
%%%%%%%%%%%%%%%%%%%%%%%%%%%%%%%%%%%%%%%%%%
Consider the 3D-point cloud $\mathcal{P}_2 = \{P^1, P^2\},$ where $P^1=(400, 9, 100)^{\top},$ and $P^2=(599, 10, 100)^{\top}$. Thus,  $r_p = 100$ and  $P^0 = (500, 10, 100)^{\top}$. Let  $K^0=( 0.7, 0.2, -0.6, 0.9, -0.8, -0.7 )^{\top}$. By using the Chebyshev map with $D=3$, the permutation map $\Pi_1 = (9, 12, 7, 11, 8, 10, 6, 1, 4, 3, 2, 5)$ and the scaling parameter $\psi=1/9$, we obtain the following cipher points
$$C^1 = \begin{pmatrix}123.6391 \\ 388.6309 \\ 575.0550 \\
\end{pmatrix}, \quad
C^2 = 
\begin{pmatrix}
151.8342 \\ -21.5335 \\ 24.5006 \\
\end{pmatrix}.
$$ 
The ciphertext $C^1, C^2$ are  far away from the given sphere as illustrated in Figure \ref{Fig:CE2}(c). Also, we notice that after the permutation, the points $P^{\prime 1}, P^{\prime 2} ,O^{\prime 1},$ and  $O^{\prime 2}$ are out of the given sphere as illustrated in Figure \ref{Fig:CE2}(b). This counterexample indicates that the cipher of \cite{jolfaei20143d} is not geometrically stable, even with  smaller $\psi$.
\end{example}
%%%%%%%%%%%%%%%%%%%%%%%%%%%%%%%%%%%%%%%%%%%
%% Figure: counterexample 2 %%
%%%%%%%%%%%%%%%%%%%%%%%%%%%%%%
\begin{figure}[!ht]
	\centering
	\begin{subfigure}[t]{.22\textwidth}
		\centering
		\includegraphics[width=\linewidth, height=4.3cm]{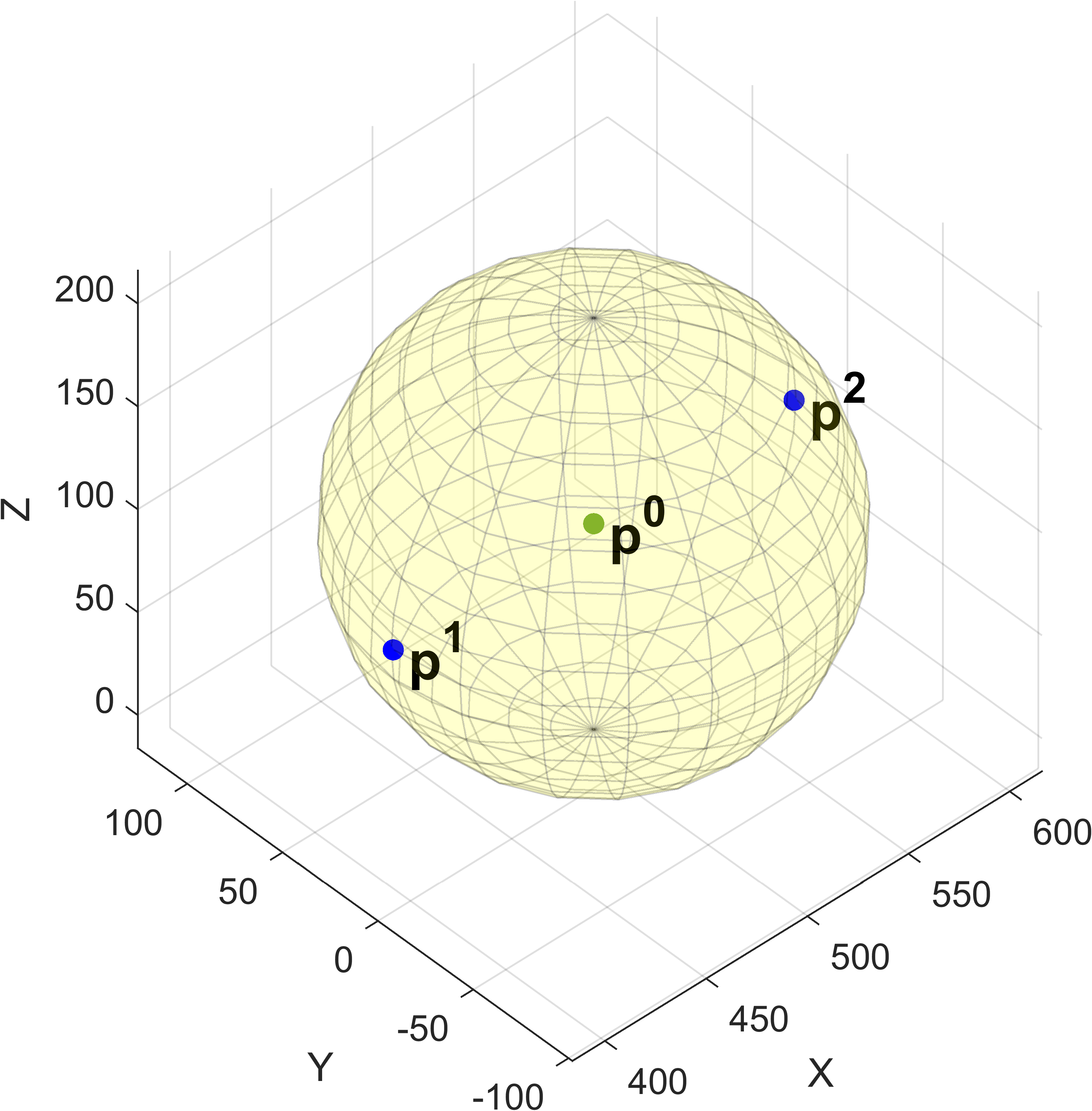}
		\caption{}\label{fig:ce2_PPC}
	\end{subfigure}
	\hspace{2mm}
	\begin{subfigure}[t]{.23\textwidth}
		\centering
		\includegraphics[width=\linewidth, height=4.3cm]{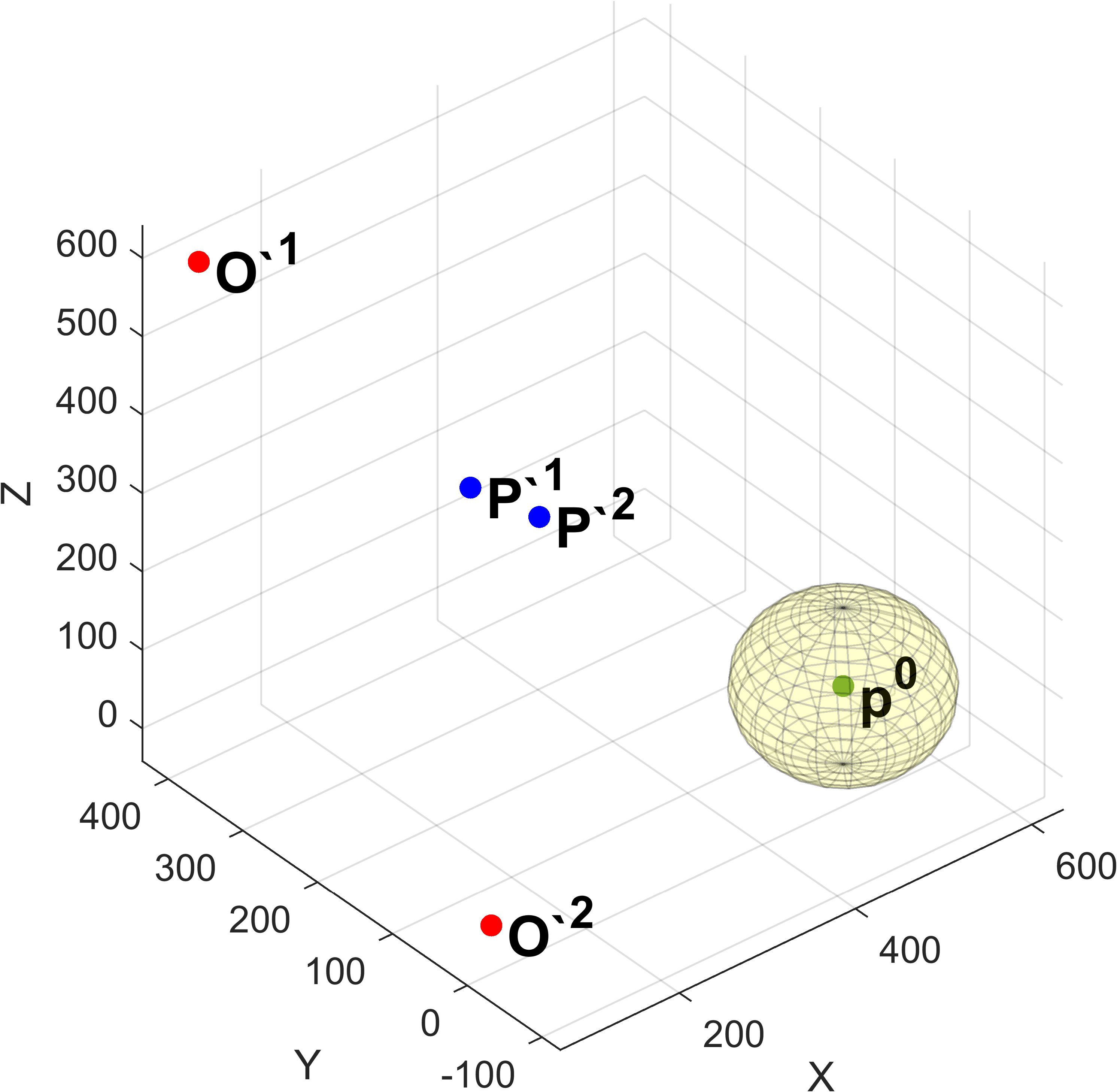}
		\caption{}\label{fig:ce2_PermutedPs}
	\end{subfigure}
	
	\medskip
	
	\begin{subfigure}[t]{.22\textwidth}
		\centering
		\vspace{0pt}% set the real top as the top
		\includegraphics[width=\linewidth, height=4.3cm]{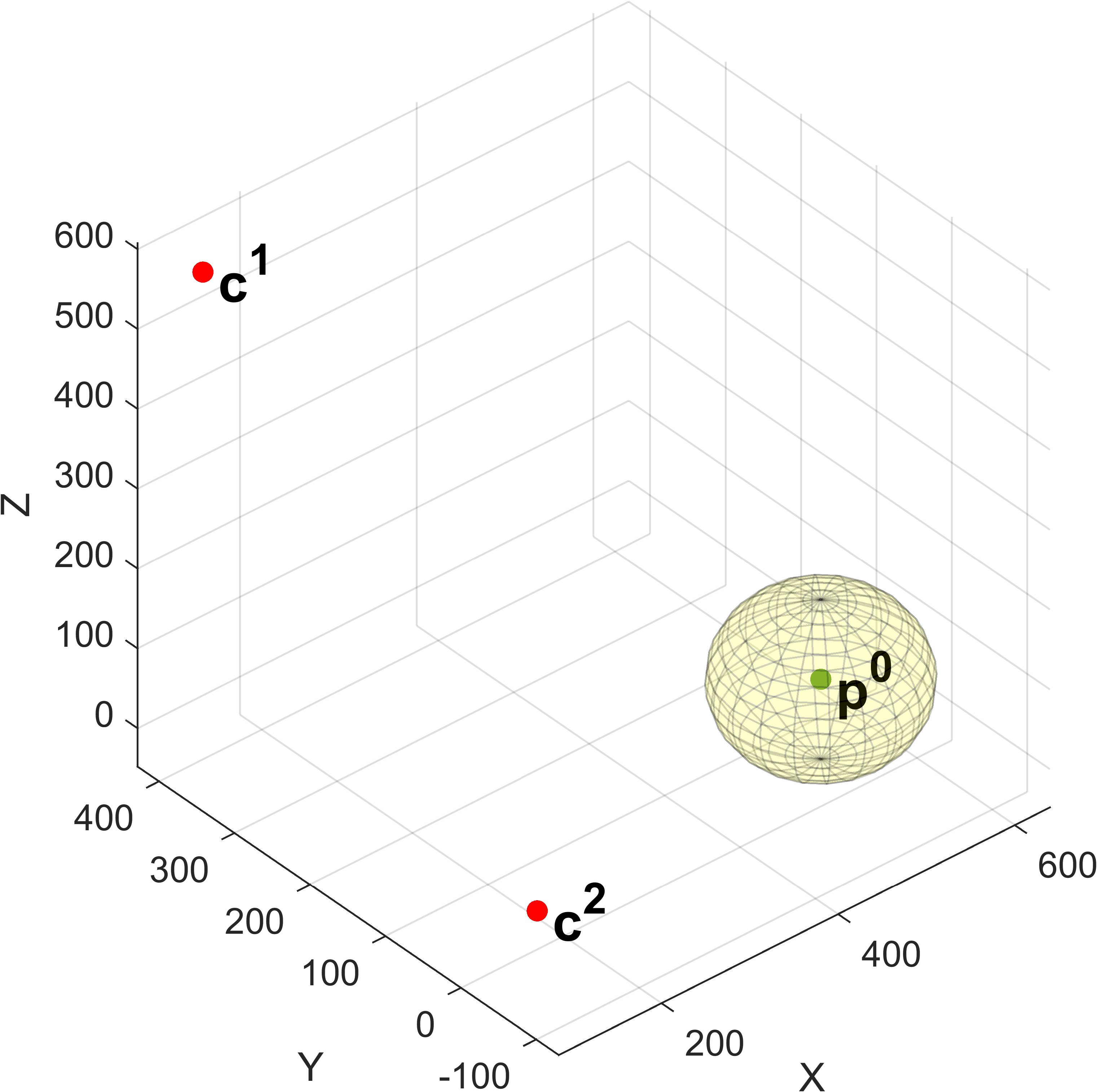}
		\caption{}\label{fig:ce2_CPC}
	\end{subfigure}
	\hspace{2mm}
	\begin{minipage}[t]{.23\textwidth}
		%\vspace{2mm}
		\caption{\small{An illustration of Example \ref{counter_ex2}. (a) Plain points $P^1$ and $P^2$. (b) Permuted Points $P^{\prime 1}, P^{\prime 2} ,O^{\prime 1} , \text{and } O^{\prime 2}$. (c) Cipher points $C^1$ and $C^2$ lie away from $P^0, r_c > r_p$.} }
		\label{Fig:CE2}
	\end{minipage}
\end{figure}

In the following, we investigate the geometric stability of encryption algorithms based on shuffling coordinates and localized rigid body rotations.  Hereafter, we follow the same notation for $\mathcal{P}_N, \mathcal{C}_N, P^0, C^0, r_p, r_c$ as mentioned above. For convenience, we assume that $C^j$ denotes $C^j_1$ and $R^j$ denotes $R^j(\alpha^j_{11},\alpha^j_{12},\alpha^j_{13}), 1 \leq j \leq N$.
%%%%%%%%%%%%%%%%%%%%%%%%%%%%%%%%%%%%%%%%%%%
%% Figure: Shuffling the coordinates ... %%
%%%%%%%%%%%%%%%%%%%%%%%%%%%%%%%%%%%%%%%%%%%
\begin{figure}[!ht]
	\centering
	\includegraphics[width=2.5in]{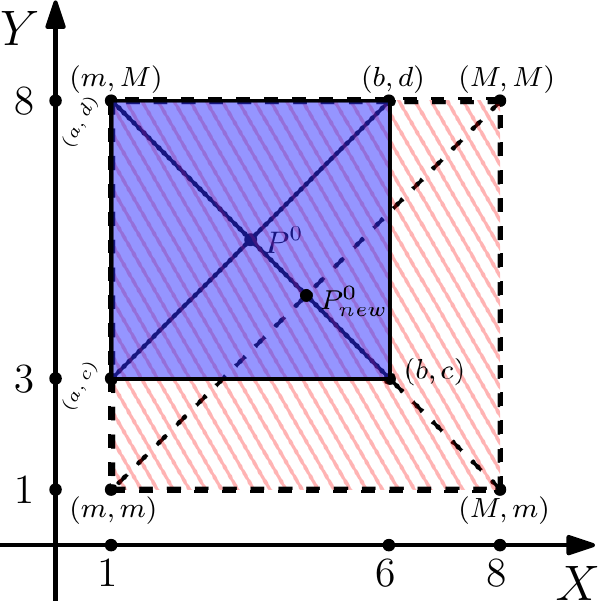}
	\caption{\small {Shuffling the coordinates of the points that lie in the blue square $R = \{(x,y)\in \mathbb{R}^2 : 1 \leq x \leq 6, 3 \leq y \leq 8\}$ will result in points that lie in the red square $R^{\prime} = \{(x,y)\in \mathbb{R}^2: 1 \leq x,y \leq 8\}$. Notice that the center $P^0 = (3.5,5.5)$ of $R$ is updated to $P^0_{new} = (4.5,4.5)$. Here $m=1, \ M=8, \ \rho = 7\sqrt{2}.$}}
	\label{Fig:Fig_100}
\end{figure}
%%%%%%%%%%%%%%%%%%%%%%%%%%%%%%%%%%%%%%%%%%%

Before deriving the mathematical proofs in $\Bbb R^3$, let us consider a 2D-setting that indicates that shuffling the coordinates of points within a certain domain may lead to a huge disorder, particularly when the center $P^0$ is not the origin. Assume that the points of $\mathcal{P}_N$ lie in a rectangle $R = \{(x,y)\in \mathbb{R}^2 : a \leq x \leq b, c \leq y \leq d\}$, centered at $P^0=(\frac{a+b}{2}, \frac{c+d}{2})$. Let $m = \min\{a, c\}$, and $M = \max\{b, d\}$. If we shuffle the coordinates of $P^1, \cdots, P^N$ to have a new set $\mathcal{P}_N^{\prime} = \{P^{\prime 1}, \cdots, P^{\prime N}\}$, then $\mathcal{P}_N^{\prime} \subseteq R^{\prime}$, where $R^{\prime} = \{(x,y)\in \mathbb{R}^2: m \leq x,y \leq M\}$ is a square centered at the updated center $P^0_{new} = (\frac{m+M}{2},\frac{m+M}{2})$. As shown in Figure \ref{Fig:Fig_100}, the farthest distance between any two points of $\mathcal{P}_N^{\prime}$ will be $\rho = \sqrt{2} (M-m)$. This is the situation when the cipher only permutes the coordinates of the plaintext.

Now we consider the case when a cipher shuffles the coordinates of the plaintext with other randomly selected points, for instance the cipher of \cite{jolfaei20143d} as shown in Figure \ref{Fig:Fig_200}. Let $\mathcal{P}_N$ lie in a circle centered at $P^0 = (p^0_1, p^0_2)$ with radius $r$ and add to $\mathcal{P}_N$ a random set $\mathcal{O}_N$ which lies in a circle centered at $P^0$ with radius $\sqrt{2}r$. The set $\mathcal{P}_N \cup \mathcal{O}_N$ lies in the square $R = \{(x,y)\in \mathbb{R}^2 : p^0_1-\sqrt{2}r \leq x \leq p^0_1+\sqrt{2}r, \ p^0_2-\sqrt{2}r \leq y \leq p^0_2+\sqrt{2}r\}$, as shown in Figure \ref{Fig:Fig_200}. Therefore, shuffling the coordinates of the points of $\mathcal{P}_N \cup \mathcal{O}_N$ will form points in the square  $R^{\prime} = \{(x,y)\in \mathbb{R}^2: m-\sqrt{2}r \leq x,y \leq M+\sqrt{2}r\}$, centered at $P^0_{new}=(\frac{m+M}{2}, \frac{m+M}{2})$, where $m = \min\{p^0_1, p^0_2\}$, and $M = \max\{p^0_1, p^0_2\}$. Thus, the farthest points among the shuffled points will be at a distance of 
\begin{equation}
\rho = 4r + \sqrt{2}\; \vert p^0_1 - p^0_2 \vert .
\end{equation}
%%%%%%%%%%%%%%%%
%% Figure: 
%%%%%%%%%%%%%%%%
\begin{figure}[!ht]
	\centering
	\includegraphics[width=2.5in]{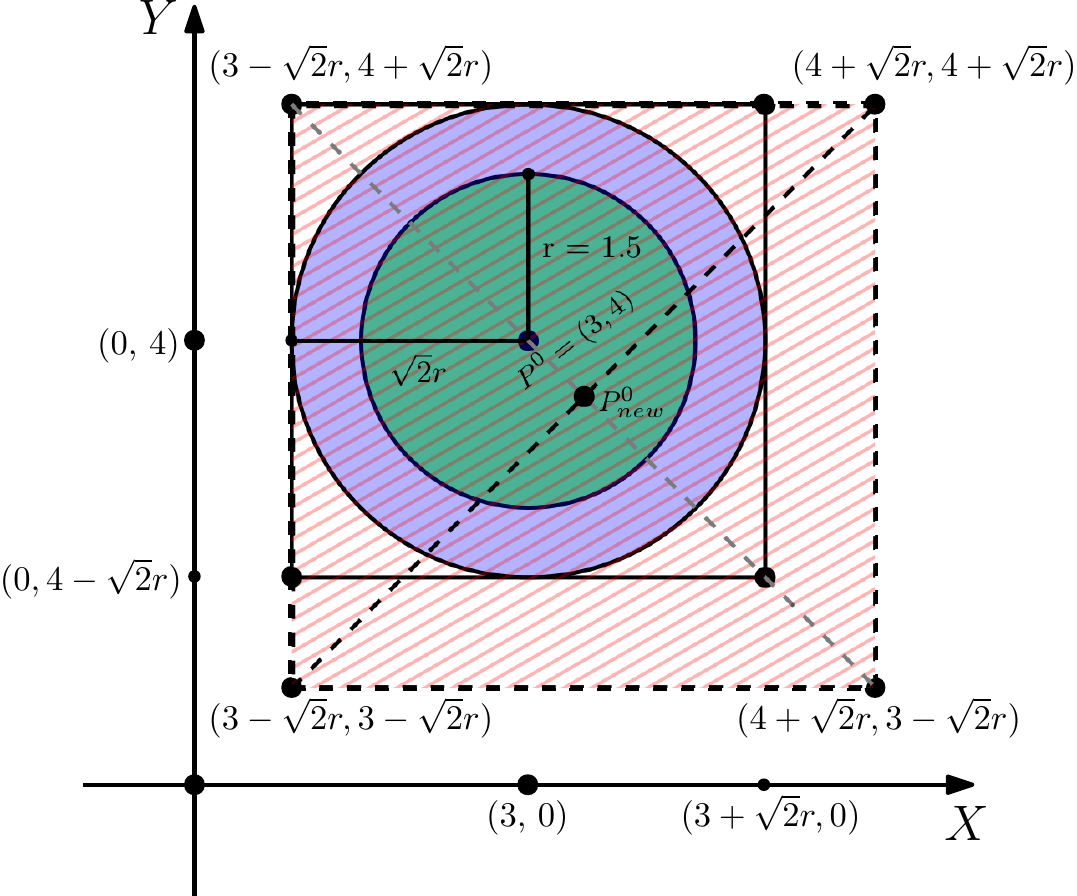}
	\caption{Points of $\mathcal{P}_N$ lie in the green circle of radius $r = 1.5$ and points of $\mathcal{O}_N$ lie in the blue circle of radius $\sqrt{2}r$. Shuffling the coordinates of $\mathcal{P}_N \cup \mathcal{O}_N$ will form points in the square $R^{\prime} = \{(x,y)\in \mathbb{R}^2: 3-\sqrt{2}r \leq x,y \leq 4+\sqrt{2}r\}$. The center $P^0 = (3,4)$ is updated to be \ $P^0_{new} = (3.5,3.5)$, and $\rho = 6 + \sqrt{2}. $}
	\label{Fig:Fig_200}
\end{figure}
%%%%%%%%%%%%%%%%
Now, we consider the 3D-setting.
%%%%%%%%%%%%%%%%
%% Figure: 
%%%%%%%%%%%%%%%%
\begin{figure}[!ht]
	\centering
	\includegraphics[width=2.5in]{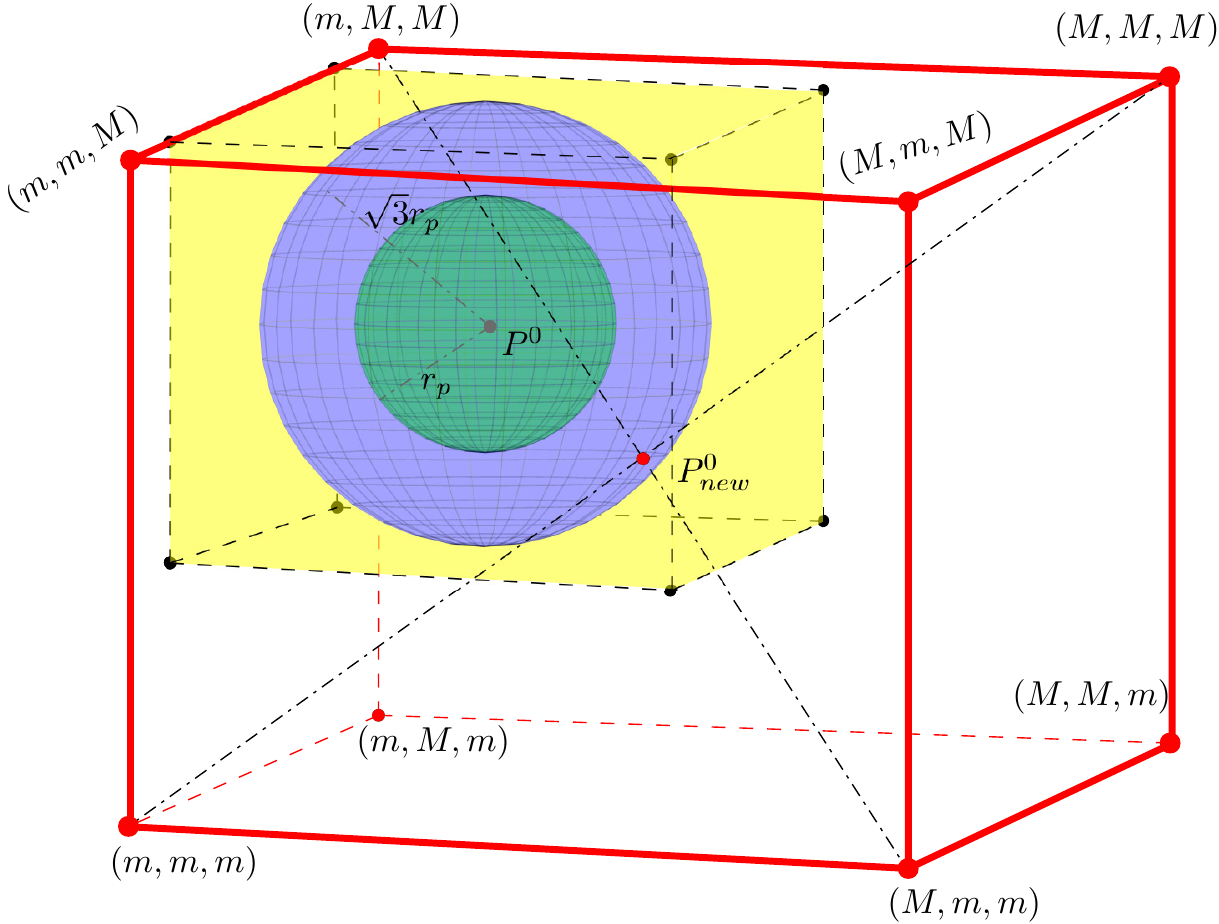}
	\caption{\small{The points of $\mathcal{P}_N$ lie in the green sphere of radius $r_p = 5$ and center $P^0$, and the points of $\mathcal{O}_N$ lie in the blue sphere of radius $\sqrt{3}r_p$ and center $P^0$. Shuffling the coordinates of the points of $\mathcal{P}_N \cup \mathcal{O}_N$ will result in points in the extended cube $R^{\prime} = \{(x,y)\in \mathbb{R}^2: m \leq x,y,z \leq M \}, M = M_0 + \sqrt{3}r_p, \ m = m_0 - \sqrt{3}r_p$. Notice that the center $P^0 = (10, 15, 20)$ is updated to the new one $P^0_{new} = (15,15,15)$ and $\rho =\sqrt{3}(M-m)= 30 + 10\sqrt{3}. $}}
	\label{fig:out_sphere_3D}
\end{figure}
%%%%%%%%%%%%%%%%
%% Lemma 1:
%%%%%%%%%%%%%%%%
\begin{lemma}\label{lem_1}
For  $j = 1, 2, \cdots, N,$ we have for $0 < \psi$, and $C^j \in \mathcal{C}_N$, generated via \textup{(\ref{eq:rot})},
\begin{equation}\label{eq:lem1_1}
\|C^j - P^0\|\leq(\psi + 1)(6 r_p + \sqrt{3} (M_0 - m_0)),
\end{equation}
where $m_0$ and $M_0$ are
\begin{equation}\label{eq:lem1_2}
m_0 := \underset{\tiny{1 \leq i\leq 3}}{\min} p^0_i, \quad  M_0 := \underset{\tiny{1 \leq i\leq 3}}{\max} p^0_i.
\end{equation}
\end{lemma}
\begin{proof}
Let $\mathcal{O}_N = \{O^1, \cdots, O^N\}$ be 3D points randomly created via (\ref{Eq: Randpseudopts}), where $\upsilon = 1$. Thus, $\| O^j - P^0 \| \leq \sqrt{3}r_p, P^0 = (p^0_1, p^0_2, p^0_3)^{\top}$. Hence, the set $\mathcal{P}_N \cup \mathcal{O}_N$ lie in the cube,
\begin{equation}\label{eq:lem1_3}
\begin{split}
\Omega = \left\{  (x^1, x^2, x^3)^\top \in \mathbb{R}^3: \vert x^i - p^0_i \vert \leq \sqrt{3}r_p \right\},
\end{split}
\end{equation}
which is depicted as the yellow cube in Figure \ref{fig:out_sphere_3D}. Then, the sets 		
\[
\mathcal{P}_N^{\prime} = \{ P^{\prime 1}, P^{\prime 2}, \cdots, P^{\prime N}\},\,
\mathcal{O}_N^{\prime} = \{ O^{\prime 1}, O^{\prime 2}, \cdots, O^{\prime N}\},
\]
resulting from shuffling the coordinates of $\mathcal{P}_N \cup \mathcal{O}_N$ will lie in the cube $\Omega^{\prime}$ of all points $ (x^1, x^2, x^3)^\top\in\Bbb R^3$, for which 
\begin{equation}\label{eq:lem1_5}
(m_0 -\sqrt{3}r_p) \leq x^i \leq (M_0 +\sqrt{3}r_p), i=1,2,3.
\end{equation}
The cube $\Omega^{\prime}$ is illustrated as the red cube in Figure \ref{fig:out_sphere_3D}. Noting that the side length of the cube $(\Omega^{\prime})$ is $2\sqrt{3}r+(M_0 - m_0)$, then for $j = 1, 2, \cdots, N$, we obtain
\begin{equation}\label{eq:lem1_6}
\| P^{\prime j} - O^{\prime j} \| \leq 6r_p + \sqrt{3} (M_0 - m_0),
\end{equation}
\begin{equation}\label{eq:lem1_7}
\| P^{\prime j} - P^0 \|, \| O^{\prime j} - P^0 \| \leq 6r_p + \sqrt{3} (M_0 - m_0).
\end{equation}
Now from (\ref{eq:rot}), and using the triangle inequality, we obtain
\begin{eqnarray*}\label{eq:lem1_8}
\| C^j - P^0 \| &\leq& \psi \|R^j\| \| P^{\prime j} - O^{\prime j} \|+ \| O^{\prime j} - P^0 \| 
\\ &\leq& \psi \cdot 1 \cdot [6r_p \ + \ \sqrt{3}(M_0 - m_0)]  
\\ &\quad& + \ 6r_p \ + \sqrt{3} (M_0 - m_0) 
\\ &=& (6 \psi + 6) r_p + (\sqrt{3} \psi + \sqrt{3}) (M_0 - m_0),
\end{eqnarray*}
which is (\ref{eq:lem1_1}). Here, we have used the fact that $\|R^j\| = 1$ since $R^j$ is orthogonal.\qed
\end{proof}

\begin{remark}
In \cite{jolfaei20143d}, $\psi$ is taken as a scaling factor that satisfies $0 < \psi \leq \frac{1}{9}$. From estimate (\ref{eq:lem1_1}), it is clear that the ciphertext $\mathcal{C}_N$ is widely deviated from the origin, apart from any choice of $\psi> 0$,  no matter however small it is taken.  Nevertheless,  the instability of the cipher of \cite{jolfaei20143d} is already proved by Example \ref{counter_ex2}. The inequality in (\ref{eq:lem1_1}) indicates the expected amount of instability, which will always be present even if $\psi$ is taken arbitrary small. The next corollary indicates also that taking $P^0$ to be the origin will not lead to geometric stability. The bound in (\ref{eq:lem1_7}) may become tighter depending on $P^{\prime j}, O^{\prime j}$. As a consequence of Lemma \ref{lem_1}, \cite[Lemma 1]{jolfaei20143d} is flawed and the subsequent results are not correct.
\end{remark}

\begin{corollary}\label{cor_1}
If $P^0$ is the origin $(0, 0, 0)^\top$, then
\begin{equation}\label{eq:cor1_1}
\| C^j - P^0 \| \leq 6(\psi + 1) r_p.
\end{equation}
\end{corollary}
\begin{proof}
The result follows by substituting $m_0 = M_0 = 0$ in (\ref{eq:lem1_1}).\qed
\end{proof}
Now, we reconsider Equation (\ref{eq:rot}) above, and modify it to be
\begin{equation}\label{eq:rotation_modified}
C^j = \psi \left[ R^j (P^{\prime j} - O^{\prime j}) + O^{\prime j} \right] , \quad 1 \leq j \leq N.
\end{equation}

\begin{lemma}\label{lem_3}
Let $C^j$ be as defined in \textup{(\ref{eq:rotation_modified})}. Then, for $j = 1, 2, \cdots, N$, we have
\begin{equation}\label{eq:lem3_1}
\| C^j - P^0 \| \leq \psi \left[ 12 r_p + 3 \sqrt{3}(M_0 - m_0) \right] + (1- \psi) \| P^0 \|.
\end{equation}
In particular, if $P^0$ is the origin, then
\begin{equation}\label{eq:lem3_2}
\| C^j - P^0 \| = \| C^j \|\leq 12 \psi r_p,
\end{equation}
i.e. the cipher will be geometrically stable if $0 < \psi \leq \frac{1}{12}$.
\end{lemma}
\begin{proof}
Using (\ref{eq:rotation_modified}) and the triangle inequality,
\begin{equation}\label{eq:lem3_3}
\| C^j - P^0 \| \leq \psi \|R^j\| \| P^{\prime j} - O^{\prime j} \|+ \| \psi O^{\prime j} - P^0 \|.
\end{equation}
Using the triangle inequality once more
\begin{equation} \label{eq:lem3_4}
\begin{aligned}[b]
\| \psi O^{\prime j} - P^0 \| &\leq \| \psi O^{\prime j} - \psi P^0 + \psi P^0 - P^0   \| \\
&\leq \psi \| O^{\prime j} - P^0 \| + (1 - \psi) \| P^0 \|.
\end{aligned}
\end{equation}
Substituting (\ref{eq:lem1_7}) in (\ref{eq:lem3_4}) and combining the result with (\ref{eq:lem3_3}) and (\ref{eq:lem1_6})  yield (\ref{eq:lem3_1}). The inequality in (\ref{eq:lem3_2}) and geometric stability follow when $P^0$ is set to $(0,0,0)^\top$ in (\ref{eq:lem3_1}).\qed
\end{proof} 
%%%%%%%%%%%%%%%%
%%%%%%%%%%%%%%%%%%%%%%%%%%%%%%%%%%%%%%%%%%%
\begin{remark} The previous proofs and Example\ref{counter_ex2} prove that  the cipher of \cite{jolfaei20143d} in not geometrically stable. However, it is worthwhile to mention that this will happen most likely when the points of $\mathcal{P}_N$ and $\mathcal{O}_N$ lie on the boundaries of their domains and based on the permutation process. We see also that the scaling factor must be  $\frac{1}{12}$, not $\frac{1}{9}$ in the case when $P^0$ is the origin and we use (\ref{eq:rotation_modified}).
\end{remark} 
%%%%%%%%%%%%%%%%%%%%%%%%%%%%%%%%%%%%%%%%%%%%%%%%%%%%%%%%%%%%%%%%%%%%%%%%%%%%%%%%%%%%%%%%%%%%%%%%%%%%%%%%%%%%%%%%%%%%%%%%%%%%%%%%%%%%%%%%%%%%%%%%%%%%%%%%%%%%%%%%%%%%%

\section{Conclusions}
\label{section:Conclusions}
This work presents a comprehensive study of geometric stability of 3D-point cloud encryption techniques, which are based on shuffling coordinates and rigid body rotation. We investigate spatial and dimensional stabilities of 3D-point cloud ciphers. We indicate through detailed investigations with mathematical proofs that ciphers, which are based on permuting 3D-point cloud coordinates with other random points may lead to a geometric disorder which is greater than what was thought in relevant literature. The technique demonstrated here is applicable to any cipher that is based on shffling coordinates and rigid body motion, apart from the faulty algorithm of \cite{jolfaei20143d}.
%%%%%%%%%%%%%%%%%%%%%%%%%%%%%%%%%%%%%%%%%%%%%%%%%%%%%%%%%%%%%%%%%%%%%%%%%%%%%%%%%%%%%%%%%%%%%%%

%\input sn-sample-bib.tex
%\clearpage
\bibliographystyle{abbrv}
{\footnotesize \bibliography{refs}}

\end{document}